\documentclass[11pt]{article}
\usepackage[a4paper, total={6in, 8in}]{geometry}
\usepackage{enumerate}
\usepackage{amsmath}
\usepackage{amssymb,latexsym}
\usepackage{amsthm}
\usepackage{color}
\usepackage{cancel}
\usepackage{graphicx}
\usepackage{cite}

\makeatletter
\usepackage{comment}
\let\wfs@comment@comment\comment
\let\comment\@undefined

\usepackage{changes}
\let\wfs@changes@comment\comment
\let\comment\@undefined

\newcommand\comment{%
    \ifthenelse{\equal{\@currenvir}{comment}}
    {\wfs@comment@comment}
    {\wfs@changes@comment}%
}

\definechangesauthor[name=Daniele, color=red]{DAN}
\definechangesauthor[name=Marco, color=blue]{MARCO}
\usepackage{lineno}
\usepackage{hyperref}
\usepackage{comment}
\usepackage{amscd}

\newtheorem{theorem}{Theorem}[section]

\newtheorem{definition}[theorem]{Definition}
\newtheorem{proposition}[theorem]{Proposition}

\newtheorem{remark}[theorem]{Remark}

\title{On a conjecture on APN permutations}
\author{Daniele Bartoli\thanks{Dipartimento di Matematica e Informatica, Universit\`a degli Studi di Perugia,  Perugia, Italy. daniele.bartoli@unipg.it}\,\, and
Marco Timpanella\thanks{Dipartimento di Matematica e Fisica, Universit\`a degli Studi della Campania ``Luigi Vanvitelli'', Caserta, Italy.
marco.timpanella@unicampania.it}
}
\date{ }

\begin{document}

\maketitle
\begin{abstract}
The single trivariate representation proposed in [C. Beierle, C. Carlet, G. Leander, L. Perrin, A Further Study of Quadratic APN Permutations in
Dimension Nine, arXiv:2104.08008] of the two sporadic quadratic APN permutations in dimension 9 found by Beierle and Leander \cite{Beierle} is further investigated. In particular, using tools from algebraic geometry over finite fields, we prove that such a family does not contain any other APN permutation for larger dimensions.
\end{abstract}


\section{Introduction}

Vectorial Boolean functions play an important role in cryptography, as they are one of the key ingredients in the design of secure cryptographic primitives. In order for
these primitives to resist to differential attacks \cite{Shamir}, vectorial Boolean functions with strong properties must be employed. One of these properties has been captured in the definition of APN functions.

\begin{definition}
A function $F:\mathbb{F}_{2}^n\rightarrow \mathbb{F}_{2}^m$, $n,m$ positive integers, is APN if for every $\alpha\in \mathbb{F}_{2}^n$, $\alpha\neq 0$, and $\beta \in \mathbb{F}_{2}^m$, the equation $F(x+\alpha)+F(x)=\beta$ has at most two solutions for $x\in \mathbb{F}_{2}^n$.
\end{definition}

Because of their applications, APN functions have been widely investigated; see for instance \cite{CKM, tanti, MR2648536, Canteaut, BCV} and the survey \cite{Delgado}. In the design of symmetric primitives, APN functions are often required to be permutations.
However, it seems that APN permutations are as rare as they are interesting, and very little is known about them.
Up to CCZ-equivalence, all of the APN permutations known so far belong  to a few families, namely:

\begin{enumerate}
    \item APN monomial functions in odd dimension;
    \item one infinite family of quadratic polynomials in odd dimension \cite{Quadratic};
    \item Dillon’s permutation in dimension $6$ \cite{Dillon};
    \item two sporadic quadratic APN permutations in dimension $9$ \cite{Beierle}.
\end{enumerate}

Note that only the first two items above are actually infinite families. 

The example  in \cite{Dillon} was found to be a particular case of a specific structure called  ``butterfly"; see \cite{Perrin}. Such a structure was later generalized, but unfortunately it was proved in \cite{Canteaut} that it is impossible for a generalized butterfly to be APN unless it operates on 6 bits.

Very recently, the two sporadic quadratic permutations  in dimension $9$ obtained in \cite{Beierle} have been investigated in \cite{Carlet}. A single trivariate representation (up to EA-equivalence) of those two permutations as 
\begin{eqnarray*}
C_u:& \mathbb{F}_{2^m}^3&\rightarrow \mathbb{F}_{2^m}^3\\
& (x,y,z)&\mapsto (x^3 + uy^2z, y^3 + uxz^2, z^3 + ux^2y)
\end{eqnarray*}
was proposed. 

In particular, for $m \geq 3$ being a multiple of 3 and $u \in \mathbb{F}_{2^m}$ not being a $7$-th power, the authors of \cite{Carlet} found that  the differential uniformity of 
$C_u$ is bounded above by $8$. Also, based on numerical experiments, they conjecture that $C_u$ is not APN if $m$ is greater than $3$; see \cite[Conjecture 1]{Carlet}.

In this note, using a connection with algebraic surfaces over finite fields and an estimate on the number of $\mathbb{F}_{2^m}$-rational points related to the Lang-Weil bound, we prove that, when $m \geq 20$ is a multiple of $3$ and  $u\in \mathbb{F}_{2^m}\setminus \{0\}$ is not a $7$-th power, the trivariate function $C_u$ is not APN. 
As by \cite[Remark 4]{Carlet} it was already observed that $C_u$ is not APN for $m\in \{6,9,12,15,18\}$, our result proves the first statement in \cite[Conjecture 1]{Carlet}.

\section{Connection with algebraic surfaces}
In the following we let $q=2^m$, where $m \geq 3$ is a multiple of $3$.
Proving that $C_u$ is not APN is equivalent to showing that the homogeneous system 
\begin{equation}\label{system}
    \begin{cases}
\alpha x^2 + \alpha^2 x + u\gamma y^2 + u\beta^2 z = 0\\
\beta y^2 + \beta^2 y + u\alpha z^2 + u\gamma^2 x = 0\\
\gamma z^2 + \gamma^2 z + u\beta x^2 + u\alpha^2 y = 0\\
    \end{cases}
\end{equation}
has at least $4$ solutions for a certain choice of $(\alpha,\beta,\gamma)\in \mathbb{F}_{q}^3$, $(\alpha,\beta,\gamma)\neq (0,0,0)$; see \cite[Theorem 2]{Carlet}.
As the authors point out  in the proof of  \cite[Theorem 2]{Carlet}, System \eqref{system} has at most $2$ solutions if $\alpha\beta\gamma=0$, and hence in the following we will assume $\alpha\beta\gamma\neq 0$. Also, we will consider $\alpha \neq u^2\beta^3/\gamma^2$.

Note that $(0,0,0)$ and $(\alpha,\beta,\gamma)$ are always solutions of System \eqref{system}.

From the first equation of System \eqref{system}, $$z=\frac{\alpha x^2+\alpha^2 x+u\gamma y^2}{u \beta^2},$$ and thus we obtain
\begin{equation}\label{system2_0}
    \begin{cases}
    r_1(x,y):=\alpha^3x^4 + \alpha^5x^2 + u^2\beta^4\gamma^2x + u^2\alpha \gamma^2y^4 + u\beta^5y^2 + u\beta^6 y=0\\
    r_2(x,y):=\alpha^2\gamma x^4 + (\alpha^4\gamma+u\alpha\beta^2\gamma^2+u^3\beta^5)x^2 + u\alpha^2\beta^2\gamma^2 x\\  \hspace{1.85 cm} +u^2\gamma^3y^4 + u^2\beta^2\gamma^3 y^2 +u^3\alpha^2\beta^4 y=0.\\
    \end{cases}
\end{equation}

Taking the linear combination $\gamma r_1(x,y)+\alpha r_2(x,y)=0$, System \eqref{system2_0} is equivalent to 

\begin{equation}\label{system2}
    \begin{cases}
    \alpha^3x^4 + \alpha^5x^2 + u^2\beta^4\gamma^2x + u^2\alpha \gamma^2y^4 + u\beta^5y^2 + u\beta^6 y=0\\
   (\alpha^2 \gamma^2+ u^2\alpha\beta^3)x^2 + (\alpha^3\gamma^2 + u\beta^2\gamma^3)x + (u\alpha \gamma^3 + \beta^3\gamma)y^2 + (u^2\alpha^3\beta^2+ \beta^4\gamma)y=0.\\
    \end{cases}
\end{equation}

Since  $(\alpha^2 \gamma^2+ u^2\alpha\beta^3)\neq 0$,  
\begin{equation}\label{x2}
x^2=\frac{(\alpha^3\gamma^2 + u\beta^2\gamma^3)x + (u\alpha \gamma^3 + \beta^3\gamma)y^2 + (u^2\alpha^3\beta^2+ \beta^4\gamma)y}{\alpha^2 \gamma^2+ u^2\alpha\beta^3},
\end{equation}
and 

\begin{eqnarray}\nonumber 
x^4&=&\frac{(\alpha^3\gamma^2 + u\beta^2\gamma^3)^2x^2 + (u\alpha \gamma^3 + \beta^3\gamma)^2y^4 + (u^2\alpha^3\beta^2+ \beta^4\gamma)^2y^2}{(\alpha^2 \gamma^2+ u^2\alpha\beta^3)^2}\\
&=&\frac{Ax+ By^4 + Cy^2 + Dy}{(\alpha^2 \gamma^2+ u^2\alpha\beta^3)^3},\label{x4}
\end{eqnarray}
where
\begin{eqnarray*}
A&=&\gamma^6(\alpha^3 + u\beta^2\gamma)^3,\\
B&=&\alpha\gamma^2(\alpha\gamma^2 + u^2\beta^3)(u\alpha\gamma^2 + \beta^3)^2,\\
C&=&(\alpha^4\gamma^2 + u^2\alpha^3\beta^3 + u\alpha\beta^2\gamma^3 + \beta^5\gamma)\cdot\\
&&\cdot(u^4\alpha^4\beta^4 + u\alpha^3\gamma^5+ u^3\alpha^2\beta^3\gamma^3 + \alpha^2\beta^3\gamma^3 + u^2\alpha\beta^6\gamma + u^2\beta^2\gamma^6),\\
D&=&\gamma^4\beta^2(\alpha^3 + u\beta^2\gamma)^2(u^2\alpha^3 + \beta^2\gamma).
\end{eqnarray*}

Using  \eqref{x2} and \eqref{x4} in System \eqref{system2} reads

\begin{equation}\label{system3}
    \begin{cases}
    u^3\beta^6\gamma^2(u\alpha^7 + u^2\alpha^4\beta^2\gamma + u\alpha^2\beta\gamma^4 + u^3\alpha\beta^4\gamma^2 + u^5\beta^7 + \gamma^7)x=Q(y)\\
   (\alpha^2 \gamma^2+ u^2\alpha\beta^3)x^2 + (\alpha^3\gamma^2 + u\beta^2\gamma^3)x + (u\alpha \gamma^3 + \beta^3\gamma)y^2 + (u^2\alpha^3\beta^2+ \beta^4\gamma)y=0,\\
    \end{cases}
\end{equation}
with
\begin{eqnarray*}
Q(y)&=&(u + 1)^2(u^2 + u + 1)^2\alpha\beta^6\gamma^2(\alpha\gamma^2 + u^2\beta^3)y^4+\beta^4(u^4\alpha^8\gamma^2 + u^6\alpha^7\beta^3 + u^5\alpha^5\beta^2\gamma^3 \\
&&+ u^4\alpha^4\beta^5\gamma +
 u\alpha^3\beta\gamma^6 + u^3\alpha^2\beta^4\gamma^4 + \alpha^2\beta^4\gamma^4 +u^5\alpha\beta^7\gamma^2+ u^2\alpha\beta^7\gamma^2 + u^3\alpha\gamma^9\\
 &&+ u^7\beta^{10} + u^2\beta^3\gamma^7)y^2+u\beta^6(u^5\alpha^7\beta^2 + u^3\alpha^4\beta^4\gamma + u^3\alpha^3\gamma^6 + \alpha^3\gamma^6 + u^2\alpha^2\beta^3\gamma^4\\
&&+ u^4\alpha\beta^6\gamma^2 + u^6\beta^9 + u\beta^2\gamma^7)y.
\end{eqnarray*}

\begin{comment}
FF := GF(2^3);
K<u,alpha,beta,gamma> := PolynomialRing(FF,4);
POL := u *alpha^7 + u^2 *alpha^4 *beta^2*gamma + u* alpha^2 *beta*gamma^4 + u^3*alpha*beta^4*gamma^2 + u^5* beta^7 + gamma^7;
Factorization(Evaluate(POL,[u^7,alpha,beta,1]));
Resultant(u^5*beta + u*alpha + 1,u^5*beta + FF.1*u*alpha + FF.1^3,beta);
\end{comment}

Note that $H(\alpha,\beta,\gamma):=u\alpha^7 + u^2\alpha^4\beta^2\gamma + u\alpha^2\beta\gamma^4 + u^3\alpha\beta^4\gamma^2 + u^5\beta^7 + \gamma^7$ factorizes as
\begin{eqnarray*}
(\xi^5\beta + \xi\alpha + \gamma)\cdot
(\xi^5\beta + \eta\xi\alpha + \eta^3\gamma)\cdot
(\xi^5\beta + \eta^2\xi\alpha + \eta^6\gamma)\cdot
(\xi^5\beta + \eta^2\xi\alpha + \eta^2\gamma)\cdot\\
(\xi^5\beta + \eta^4\xi\alpha + \eta^5\gamma)\cdot
(\xi^5\beta + \eta^5\xi\alpha + \eta\gamma)\cdot
(\xi^5\beta + \eta^6\xi\alpha + \eta^4\gamma),
\end{eqnarray*}
where $\mathbb{F}_8^*=\langle\eta\rangle$ and $\xi=\sqrt[7]{u}\notin \mathbb{F}_{q}$, since $u$ is not a $7$-th power by assumption. Thus, there are no $(\alpha,\beta,\gamma) \in \mathbb{F}_q^3$, $(\alpha,\beta,\gamma)\neq (0,0,0)$ such that $H(\alpha,\beta,\gamma)=0$.

Now, after taking the resultant of the two equations of System \eqref{system3} and eliminating $x$, we are left with a polynomial in $y$ of degree $8$, namely
\begin{eqnarray*}
\bar{P}_{\alpha,\beta,\gamma}(y)&=&u\alpha\beta^{10}(\alpha\gamma^2 + u^2\beta^3)^3y(y + \beta)P_{\alpha,\beta,\gamma}(y),
\end{eqnarray*}
where 
\begin{eqnarray}\label{eqvarieta}
P_{\alpha,\beta,\gamma}(y)&=&A_6y^6+A_5y^5+A_4y^4+A_3y^3+A_2y^2+A_1y+A_0,
\end{eqnarray}
and
\begin{eqnarray*}
A_6&=& (u + 1)^4(u^2 + u + 1)^4\alpha^2\beta^4\gamma^4,\\
A_5&=&(u + 1)^4(u^2 + u + 1)^4\alpha^2\beta^5\gamma^4,\\
A_4&=&(u + 1)^4(u^2 + u + 1)^4\alpha^2\beta^6\gamma^4,\\
A_3&=&(u + 1)^4(u^2 + u + 1)^4\alpha^2\beta^7\gamma^4,\\
A_2&=&u^2(u^6\alpha^{14}+ u^2\alpha^8\beta^4\gamma^2 + \alpha^4\beta^2\gamma^8 + u^8\alpha^3\beta^5\gamma^6 + u^2\alpha^3\beta^5\gamma^6 + u^{10}\alpha^2\beta^8\gamma^4\\
&&+ u^7\alpha\beta^4\gamma^9 +u\alpha\beta^4\gamma^9 + u^8\beta^{14}u^8 + u^9\beta^7\gamma^7+ u^3\beta^7\gamma^7 + u^4\gamma^{14}),\\
A_1&=&\beta A_2,\\
A_0&=&u^4(u + 1)(u^2 + u + 1)\alpha^2\beta^3\gamma^4(u\alpha^7 + u^2\alpha^4\beta^2\gamma + u\alpha^2\beta\gamma^4 + u^3\alpha\beta^4\gamma^2 + u^5\beta^7 + \gamma^7).
\end{eqnarray*}
Note that $u^2+u+1\neq 0$. Indeed, if $m$ is odd and $u^2+u+1=0$ then $u\in \mathbb{F}_4$, a contradiction. On the other hand if $m$ is even then $u^3=1$ and, by $3\mid (2^m-1)/7$,  $u^{(2^m-1)/7}=1$, a contradiction to our assumptions on $u$.

\begin{remark}
Observe that this approach also shows that $C_u$ is a differentially $d$-uniform function with $d\leq 8$, as stated in \cite[Theorem 2]{Carlet}. Indeed, the above computations prove that this is true if $\alpha \neq u^2\beta^3/\gamma^2$. On the other hand, assume that $\alpha = u^2\beta^3/\gamma^2$. Then the second equation of System \eqref{system2} reads
$$
u\gamma^2(u^5\beta^7 + \gamma^7)x+\beta(u^3\gamma^7 y^2 + \gamma^7 y^2 + u^8\beta^8y + \beta\gamma^7 y)=0,
$$
where $u^5\beta^7 + \gamma^7\neq 0$ as $u$ is not a $7$-th power. Thus, after eliminating $x$, System \eqref{system2} is equivalent to

\begin{equation*}
\begin{cases}
x=\frac{\beta(u^3\gamma^7 y^2 + \gamma^7 y^2 + u^8\beta^8y + \beta\gamma^7 y)}{u\gamma^2(u^5\beta^7 + \gamma^7)}\\
Q_1(y)=0,
    \end{cases}
\end{equation*}
where 
\begin{eqnarray*}
    Q_1(y)&=& u^{10}\gamma^4\beta^{19} y \cdot (y + \beta)\cdot ( (u + 1)^4(u^2 + u + 1)^4\beta^{10}\gamma^{28}(y^6+\beta y^5+\beta^{2}y^4+\beta^{3}y^3)+ \\
    && u^2(u^5\beta^7 + \gamma^7)^2(u^{10}\beta^{14} + u^5\beta^7\gamma^7 + u^2\beta^7\gamma^7 + \gamma^{14})^2(y^2+\beta y)+\\
    &&u^4(u + 1)(u^2 + u + 1)(u^5\beta^7 + \gamma^7)^3\beta^9\gamma^{14}).
\end{eqnarray*}
As the degree of $Q_1(y)$ is $8$, System \eqref{system2} has at most $8$ solutions if $\alpha = u^2\beta^3/\gamma^2$. Finally, if $\alpha\beta\gamma=0$ System \eqref{system2} has at most $2$ solutions (see the first part of the proof of \cite[Theorem 2]{Carlet}), and hence $C_u$ is  differentially $d$-uniform with $d\leq 8$.
\end{remark}
\begin{comment}
K<x,y,z,a,b,c,u> := PolynomialRing(GF(2),7);

 F1 := a*x^2 + a^2*x + u*c*y^2 + u*b^2*z;

 F2 :=b*y^2 + b^2*y + u*a*z^2 + u*c^2*x;
 F3 :=c*z^2 + c^2*z + u*b*x^2 + u*a^2*y;

R1:=Resultant(F1,F2,z) div u;
R2:=Resultant(F1,F3,z);
R3:=c*R1+a*R2;
pol1:=a*c^2 + b^3*u^2;
//annullo il coeff di x^2

RR3:=Resultant(pol1,R3,a);
CC:=Coefficients(RR3,x);
//poichè u non è potenza settima il coeff di x non si annulla.
RR1:=Resultant(R1,pol1,a);
RR:=Resultant(RR3,RR1,x);
//polinomio di grado 8 e quindi è 8 differentially uniform.
\end{comment}

As a polynomial in the variables $\alpha,\beta,\gamma,y$, $P_{\alpha,\beta,\gamma}(y)$ defines a surface $V$ of degree $16$ embedded in the three-dimensional projective space $\mathrm{PG}(3,\overline{\mathbb{F}}_2)$. Good references for a more comprehensive introduction
to algebraic varieties and curves are \cite{Hartshorne,HKT}. For a survey on the use of algebraic varieties over finite fields in polynomial problems, we refer to \cite{Survey}.

Recall that System \eqref{system} always possesses the solutions $(0,0,0)$ and $(\alpha,\beta,\gamma)$. Therefore, in order to prove that $C_u$ is not APN, it is enough to exhibit at least a choice of $(\alpha,\beta,\gamma)$ for which $V$ has an $\mathbb{F}_q$-rational point not lying on $y=0$ and $y=\beta$; see Theorem \ref{th:main}.
To prove the existence of such a point, we use the following results.

The following is a particular case of \cite[Lemma 2.1]{MR2648536}.
\begin{proposition}\label{criterio2}
Let $H$ be a plane of $\mathrm{PG}(3,\overline{\mathbb{F}}_2)$ such that $V\cap H$ contains
a non-repeated absolutely irreducible component defined over $\mathbb{F}_q$. Then $V$ possesses a non-repeated absolutely irreducible component defined over $\mathbb{F}_q$.
\end{proposition}

\begin{proposition}\label{Prop:component}
There exists an $\mathbb{F}_q$-rational component of $V$ distinct from $y=0$ and $y=\beta$.
\end{proposition}
\proof
First note that $y=0$ and $y=\beta$ are not components of $V$:  it is enough to observe that $P_{\alpha,\beta,\gamma}(0)$ and $P_{\alpha,\beta,\gamma}(\beta)$ (seen as polynomials in $\alpha,\beta,\gamma,y$) are not the zero polynomial (this is readily seen by a direct computation).

Consider now the curve $\mathcal{C}$ defined as the intersection of $V$ with the plane of equation $\gamma=0$. By direct computation, this curve has homogeneous equation
$$
u^8(\alpha^7 + u\beta^7)^2 y (y + \beta)=0.
$$
The component $y=0$ is  $\mathbb{F}_q$-rational, absolutely irreducible, and non-repeated.  Then Proposition \ref{criterio2} yields the existence of an $\mathbb{F}_q$-rational component of $V$ through the line $y=0=\gamma$, which is therefore distinct from both $y=0$ and $y=\beta$. 
\endproof

To ensure the existence of a suitable $\mathbb{F}_q$-rational point of $V$, we report the following result.
\begin{theorem}\cite[Theorem 7.1]{MR2206396}\label{Th:CafureMatera}
Let $\mathcal{V}\subset\mathrm{AG}(n,\mathbb{F}_q)$ be an absolutely irreducible variety defined over $\mathbb{F}_q$ of dimension $r>0$ and degree $\delta$. If $q>2(r+1)\delta^2$, then the following estimate holds:
$$|\#(\mathcal{V}\cap \mathrm{AG}(n,\mathbb{F}_q))-q^r|\leq (\delta-1)(\delta-2)q^{r-1/2}+5\delta^{13/3} q^{r-1}.$$
\end{theorem}

\begin{theorem}\label{th:main}
If $m\geq20$, $C_u$ is not APN.
\end{theorem}
\begin{proof}
By Proposition \ref{Prop:component}, the surface $V$ contains an absolutely irreducible component $W$ defined over $\mathbb{F}_q$ of degree at most  $16$.

Since $m\geq 20$ the surface $W$ contains at least $48q$ $\mathbb{F}_q$-rational points with $\gamma=1$ (it is enough to apply Theorem \ref{Th:CafureMatera} to the dehomogenization $W_*$ of $W$ with respect to $\gamma$). 

Affine $\mathbb{F}_q$-rational points of $V_*$ lying on $\alpha\beta y(y+\beta)=0$ are contained in the three lines 
$$\alpha=0=y, \quad \alpha=0=y+\beta, \quad \beta=0=y.$$
The intersection between $V$ and $\alpha+u^2\beta^3=0$ is a degree-44 curve which has at most $44q+1$ $\mathbb{F}_q$-rational points; see \cite{HK}. 
Therefore there exists at least an $\mathbb{F}_q$-rational point $(\overline{\alpha},\overline{\beta},1,\overline{y})$ of $V_*$ with $\overline{\alpha}\overline{\beta}\overline{y}(\overline{y}+\overline{\beta})(\overline{\alpha}+u^2\overline{\beta}^3)\neq 0$. This provides a root $\overline{y}\notin \{0,\overline{\beta}\}$ of $P_{\overline{\alpha},\overline{\beta},1}(y)$ (see \eqref{eqvarieta}),  and so $C_u$ is not APN.
\end{proof}

\section*{Acknowledgements}
The research of D. Bartoli and M. Timpanella was partially supported by the Italian National Group for
Algebraic and Geometric Structures and their Applications (GNSAGA - INdAM).
The authors  are grateful to C. Beierle, C. Carlet, G. Leander, and L. Perrin for a number
of valuable comments on an earlier draft.

\end{document}